\documentclass{llncs}
\usepackage{amsfonts}
\usepackage{amsmath}
\usepackage{amssymb}
\usepackage[all]{xy}

\newtheorem{teo}{Theorem}
\newtheorem{prop}{Proposition}

\newtheorem{rem}{Remark}
\newtheorem{lem}{Lemma}
\newtheorem{defi}{Definition}
\newcommand{\M}{{\mathbb M}}
\newcommand{\am}{A^\M}
\newcommand{\az}{A^\Zset}
\newcommand{\Zset}{{\mathbb Z}}

\newcommand{\Nset}{{\mathbb N}}

\newcommand{\lang}{\mathcal L}
\newcommand{\length}[1]{\left|#1\right|}
\newcommand{\card}[1]{\left|#1\right|}

\newcommand{\et}{\textrm{ and }}

\newcommand{\xpr}[1]{"#1"}
\usepackage{stmaryrd} 
\newcommand{\abs}[1]{\left|#1\right|}
\newcommand{\soit}[1]{\left|\everymath{\displaystyle\everymath{}}\begin{array}{ll}#1\end{array}\right.}
\newcommand{\both}[1]{\left\{\everymath{\displaystyle\everymath{}}\begin{array}{l}#1\end{array}\right.}

\newcommand{\co}[2]{\left\llbracket #1,#2\right\llbracket}
\newcommand{\cc}[2]{\left\llbracket #1,#2\right\rrbracket}
\newcommand{\oo}[2]{\left\rrbracket #1,#2\right\llbracket}
\newcommand{\oc}[2]{\left\rrbracket #1,#2\right\rrbracket}
\newcommand{\ci}[1]{\co{#1}\infty}
\newcommand{\io}[1]{\oo{-\infty}{#1}}
\newcommand{\oi}[1]{\oo{#1}\infty}
\newcommand{\ic}[1]{\oc{-\infty}{#1}}
\newcommand{\scc}[2]{_{\cc{#1}{#2}}}
\newcommand{\sco}[2]{_{\co{#1}{#2}}}
\newcommand{\soo}[2]{_{\oo{#1}{#2}}}
\newcommand{\sci}[1]{_{\ci{#1}}}
\newcommand{\sio}[1]{_{\io{#1}}}
\newcommand{\soi}[1]{_{\oi{#1}}}
\newcommand{\sic}[1]{_{\ic{#1}}}
\newcommand{\sett}[2]{\left\{\left.#1\vphantom{#2}\right|#2\right\}}
\newcommand{\set}[3]{\sett{#1\in#2}{#3}}
\newcommand{\ie}{\textit{i.e.}\ }
\newcommand{\ipart}[1]{\left\lfloor #1\right\rfloor}

\title{Zigzags in Turing machines\thanks{This work has been supported by ECOS-Sud Project and CONICYT FONDECYT \#1090568.}}
\author{Anah\'i Gajardo\inst1 \and Pierre Guillon\inst2}
\institute{Departamento de Ingenier\'ia Matem\'atica, Universidad de Concepci\'on,
Casilla 160-C, Concepci\'on, Chile
\email{anahi@ing-mat.udec.cl}
\and
DIM - CMM, UMI CNRS 2807, Universidad de Chile, Av. Blanco Encalada 2120,
Santiago, Chile
\email{pguillon@dim.uchile.cl}
}

\begin{document}

\maketitle

\begin{abstract}
We study one-head machines through symbolic and topological dynamics.
In particular, a subshift is associated to the system, and we are interested in its complexity in terms of realtime recognition.
We emphasize the class of one-head machines whose subshift can be recognized by a deterministic pushdown automaton.
We prove that this class corresponds to particular restrictions on the head movement, and to equicontinuity in associated dynamical systems.
\end{abstract}


We study the dynamics of a system consisting in a finite automaton (the head) that can write and move over an infinite tape, like a Turing machine.
We use the approach of symbolic and topological dynamics.
Our interest is to understand its properties and limitations, and how dynamical properties are related to computational complexity.

This approach was initiated by K\r urka in \cite{Kurk} with two different topologies: one focused on the machine head, and the other on the tape.
The first approach was further developed in~\cite{Nich,Opro}.
More recently, in \cite{Gaja07,GajaJAC}, a third kind of dynamical system was associated to Turing machines, taking advantage of the following specificity: changes happen only in the head position whilst the rest of the configuration remains unaltered.
The whole evolution can therefore be described by the sequence of states taken by the head and the symbols that it reads.
This observation actually yields a factor map between K\r urka's first dynamical system and a one-sided subshift.

In \cite{Gaja07}, it has been proved that machines with a sofic subshift correspond to machines whose head makes only bounded cycles.
We prove here a similar characterization of machines with a shift that can be recognized by a deterministic pushdown automaton.
Moreover, we establish links between these two properties and equicontinuity in all three spaces.

In the first section, we recall the definitions and fundamental results.
The second section is devoted to defining the different dynamical systems associated to one-head machines, and to stating basic results about equicontinuity within these systems.
In the last section, we define the class of bounded-zigzag machines and state our main results.
\section{Preliminaries}
Consider a finite alphabet $A$, and $\M$ to stand either for $\Nset$ or for $\Zset$.
For a finite word $u\in A^*$, we will note $\length u$ its length, and index its letters from $0$ to $\length u-1$, unless specified otherwise.
We denote $A^{\le m}$ the set of words on $A$ of length at most $m\in\Nset$.
If $i,j\in\Zset$ and $i\le j$, $\cc ij$ will denote the closed interval of integers $i,\ldots,j$, $\co ij=\cc i{j-1}$, etc.
A point $x\in\am$ will be called \emph{configuration}.
For a configuration or a word $x$, we define $x\scc ij=x_i\ldots x_j$.
$A\sqcup B$ will denote the disjoint union of two sets $A$ and $B$.
\subsection{Topological dynamics}\label{sec:top}
A \emph{dynamical system} (DS) is a pair $(X,F)$ where $X$ is a metric space and $F$ a continuous self-map of $X$. Sometimes the space will be implicit.

The orbit of a point $x\in X$ is the set of the $F^t(x)$ for all \emph{iteration} ${t>0}$.
A point $x$ is called \emph{preperiodic} if there exist two naturals $q$, $p$ such that $F^{q+p}(x)=F^q(x)$.
If $q$ and $p$ are minimal, then $q$ is called the \emph{transient} and $p$ the \emph{period}.
When $t=0$, $x$ is called \emph{periodic}.

A point $x\in X$ is \emph{isolated} if there is an $\varepsilon>0$ such that the ball of radius $\varepsilon$ and center $x$ contains only $x$.
A point $x\in X$ is \emph{equicontinuous} for $F$ if, for any $\varepsilon>0$, there exists some $\delta>0$ such that, for any $y\in X$ with $d(x,y)<\delta$, we have that, for all $t\in\Nset$, $d(F^t(x),F^t(y)) < \varepsilon$.
The DS $(X,F)$ is \emph{equicontinuous} if, for any $\varepsilon>0$, there exists some $\delta>0$ such that, for any $x,y\in X$ with $d(x,y)<\delta$, we have that, for all $t\in\Nset$, $d(F^t(x),F^t(y)) < \varepsilon$.
When $X$ is compact, this is equivalent to having only equicontinuous points.
The DS $(X,F)$ is \emph{almost equicontinuous} if it has a residual set of equicontinuous points.

A DS $(X,F)$ is a \emph{factor} of a DS $(Y,G)$ if $\phi G=F\phi$ for some continuous onto map $\phi:Y\to X$, then called a \emph{factor map}.
\subsection{Subshifts}
We can endow the space $\am$ of {configurations} with the product of the discrete topology of $A$.
It is based on the cylinders $[u]_i=\set x\am{x\sco i{i+k}=u}$, where $i\in\M$, $k\in\Nset$ and $u\in A^k$; this notation shall be extended to semi-infinite words.
If $\M=\Zset$, $u\in A^{2r+1}$ and $r\in\Nset$, we note $[u]=[u]_{-r}$.

This topology corresponds to the metric $d:x,y\mapsto2^{-\min_{x_i\ne y_i}\abs i}$.
In other words, $d(x,y)\le 2^{-i} \Leftrightarrow x\scc{-i}i=y\scc{-i}i$; two points are \xpr{close to each other} if they coincide \xpr{around position 0}.
It is easy to extend this metric to spaces $\am\times Q$ and $\am\times Q\times\Zset$.
In that setting, $\am$ and $\am\times Q$ are compact, but $\am\times Q\times\Zset$ is not.

The \emph{shift} map is the function $\sigma:\am\to\am$ defined by $\sigma(x)_i=x_{i+1}$.
A \emph{subshift} $\Sigma$ is a closed subset of $\am$ which is also invariant by $\sigma$. It can be seen as a compact DS where the map is $\sigma$.

A subshift $\Sigma$ is characterized by its \emph{language}, containing all finite patterns that appear in some of its configurations: $\lang(\Sigma)=\sett{z\sco ij}{z\in\Sigma\et i,j\in\M}$. We denote $\lang_n(\Sigma)=\lang(\Sigma)\cap A^n$.
If the language $\lang(\Sigma)$ is regular, then we say that $\Sigma$ is \emph{sofic}.
Equivalently, a sofic subshift can be seen as the set of labels of infinite paths in some finite arc-labeled graph; this graph basically corresponds to the finite automaton that recognizes its language, without initial nor terminal state.

Any subshift can also be defined from a set of forbidden finite patterns $\mathcal F\subset A^*$ by $\Sigma=\set z\am{\forall i,j\in\M,z\sco ij\notin\mathcal F}$.
If $\mathcal F$ can be chosen to be finite, then $\Sigma$ is a subshift \emph{of finite type} (SFT).

A DS $F$ on $\am$ is completely determined by the family of its factor subshifts, \ie the factors which are also subshifts in some alphabet.
Up to some letter renaming, all factor subshifts of $F$ are of the form $(\mathcal P(F^j(x)))_{j\in\Nset}$, where $\mathcal P$ is a finite partition of $X$ into closed open sets, and $\mathcal P(y)$ denotes the unique element of this partition which contains $y\in X$.
%
\subsection{Deterministic pushdown automata}
\begin{defi}
A \emph{deterministic pushdown automaton} (DPDA) is a tuple\break $(A,\Omega,\Gamma,\bot,\lambda,o_0,F)$ where $A$ is the \emph{input alphabet}, $\Omega$ is the \emph{set of states }, $\Gamma$ is the \emph{stack alphabet}, $\bot\in\Gamma$ is the \emph{stack bottom}, $o_0$ is the \emph{initial state}, $F\subset\Omega$ is the \emph{subset of terminal states } and $\lambda:A\times\Omega\times\Gamma\to\Omega\times\Gamma^{\le2}$ is the \emph{transition function} such that: if $\lambda(a,o,\bot)=(o',\mu)$, then $\mu$ contains exactly one $\bot$, which is on its end, and if $\lambda(a,o,\beta)=(o',\mu)$ with $\beta\ne\bot$, then $\mu$ does not contain any $\bot$.

An (infinite) arc-labeled graph $G$ is associated to the automaton.
Its set of vertices is $\Omega\times(\Gamma\setminus\{\bot\})^*\bot$, and there exists an arc from $(e,\mu)$ to $(f,\nu)$ labeled $a$ if and only if $\nu=\rho\mu\sco1{\length\mu-1}$ and $\lambda(a,e,\mu_0)=(f,\rho)$. The word $\mu$ is called the \emph{stack content}.

The language $L$ recognized by the automaton consists of all words $w$ in $A^*$ such that there exists a finite path in $G$ with label $w$, starting on vertex $(o_0,\bot)$ and ending in some vertex $(o,\mu)$ with $o\in F$. A subshift is recognized by the automaton if its language is recognized by the automaton.
\end{defi}
\newcommand{\pile}[2][]{(o_{#2},#1\mu^{#2})}
\newcommand{\pilo}[2][]{(o_{#2},#1\mu_0^{#2})}
\section{Turing Machines}
In this article, a Turing Machine (TM) is a triple $(A,Q,\delta)$, where $A$ and $Q$ are the finite \emph{tape alphabet} and \emph{set of state}, and $\delta:A\times Q\to A\times Q\times\{-1,1\}$ the \emph{rule}.
We do not particularize any halting state.
We can see the TM as evolving on a bi-infinite tape.
The phase space is $X=\az\times Q\times\Zset$.
Any element of $X$ is called a configuration and represents the state of the tape, the state of the head and its position.
We consider here the topology introduced in Section \ref{sec:top}.
Thus, the farther the head is from the center, the less important become the read symbols, but the head state and position remain important.
On this (non-compact) space, $T:X\to X$ by $T(x,q,i)=(x\sio iax\soi i,p,i+d)$ if $\delta(x_i,q)=(a,p,d)$ gives the corresponding DS.
We can extend the shift function to TM configurations by $\sigma:(x,q,i)\mapsto(\sigma(x),q,i-1)$, and it clearly commutes with $T$.

We can represent the head state and position by adding a ``mark'' on the tape.
If we want a compact space, this corresponds to the following phase space:
\[X_H=\set x{(A\sqcup(A\times Q))^\Zset}{\card{\set i{\Zset}{x_i\in A\times Q}}\le1}\]
where the head position is implicitly given by the only cell with a symbol in $(A\times Q)$, and the function $T_H: X_H\longrightarrow X_H$ is defined by $T_H(x\sio i(b,q)x\soi i)=y\sio{i+d}(y_{i+d},p)y\soi{i+d}$, where $y=x\sio iax\soi i$ and $\delta(b,q)=(a,p,d)$, and $T_H(x)=x$ if $x$ does not contain any symbol in $A\times Q$.
With the topology of $X_H$ as a subshift of $(A\sqcup(A\times Q))^\Zset$, the head state and movement are less important when the head is far from $0$.
This model corresponds to the TM \emph{with moving head} defined by K\r urka in \cite{Kurk}, which highlights the tape configuration. It is a particular case of cellular automaton, \ie based on some uniformly-applied local rule.
We can intuitively see a continuous injection $\Phi:X\to X_H$ such that $\Phi T=T_H\Phi$ and $\Phi\sigma=\sigma\Phi$.

Focusing on the movements and states of the head, \cite{Kurk} also defines the system \emph{with moving tape} $T_T:X_T\to X_T$ on the (compact) space\break $X_T=\az\times Q$ by $T_T(x,q)=(\sigma^d(x\sio0ax\soi0),p)$ if $\delta(x_0,q)=(a,p,d).$
Here the head is assumed to be always at position $0$, and the tape is shifted at each step according to the rule.
There is a continuous non-injective surjection $\Psi:X\to X_T$ such that $\Psi T=T_T\Psi$.

Finally, we can have a vision centered on the head and which emphasizes only the relevant part of the configuration, as in \cite{Gaja07,GajaJAC}.
The system $S_T$ is the one-sided subshift on alphabet $Q\times A$, which is the image of the factor map $\tau_T:X_T\to S_T$ defined by $\tau_T(x,q)_t=(y_0,p)$ if $(y,p)=T_T^t(x,q)$.
In other words, it represents the sequence of pairs corresponding to the successive states of the head and the letters that it reads.
{
\[
\xymatrix{
   X_H \ar[d]_{T_H}& X \ar@{>->}[l]_\Phi \ar[d]_T \ar@{->>}[r]^\Psi& X_T \ar[d]_{T_T} \ar@{->>}[r]^{\tau_T} & S_T \ar[d]_\sigma
   \\
   X_H & X \ar@{>->}[l]_\Phi \ar@{->>}[r]^\Psi & X_T \ar@{->>}[r]^{\tau_T} & S_T
 }
\]
}
Similarly, we will note $S_H$ the one-sided subshift on alphabet $Q\sqcup(A\times Q)$ which is the image of the factor map $\tau_H:X_H\to S_H$ defined by $\tau_H(x)_t=T_H^t(x)_0$. Unlike $S_T$, this subshift does not always contain the relevant information, since the head can be completely absent.
\subsection{Equicontinuous configurations}\label{ss:eqpt}
Topological notions can actually formalize various types of head movements.
One first example is equicontinuity of the DS $T_T$.
It is strongly related with periodicity, as the next remark establishes.
This is natural since the symbol that the head reads in $X_T$ is always at position $0$.
Hence, if the head visits an infinite number of cells, say to the right, any perturbation on the initial configuration will get to position $0$, and thus will become largely significant for this topology.
We conclude the following. 
\begin{rem}\label{r:eqpt}
Let $x\in X$ be a configuration and $T$ a machine over $X$. The following statements are equivalent:
\begin{enumerate}
\item The head position on $x$ is bounded.
\item $x$ is \textbf{preperiodic} for $T$.
\item $\Phi(x)$ is \textbf{preperiodic} for $T_H$.
\item $\Psi(x)$ is \textbf{equicontinuous} for $T_T$.
\item $\tau_T\Psi(x)$ is \textbf{preperiodic} and \textbf{isolated} --\ie \textbf{equicontinuous}-- in $S_T$.
\end{enumerate}
Moreover, if one of the above occurs, then $\Psi(x)$ is preperiodic for $T_T$, $x$ is equicontinuous for $T$ and $\Phi(x)$ is equicontinuous for $T_H$.
The set of equicontinuous configurations for $T_T$ is a union of cylinders of $X_T$.
\end{rem}
If $\Psi(x)$ is preperiodic for $T_T$, then $\tau_T\Psi(x)$ is also periodic (for $\sigma$), but $x$ need not be periodic for $T$.
For example, a machine that simply moves to the left on every symbol will produce a periodic point for $T_T$ if the initial configuration $x$ is spatially periodic.
From the previous remark, such a point is not equicontinuous, and $\tau_T\Psi(x)$ is a non-isolated periodic point in $S_T$, because any perturbation of $x$ will produce a neighbor of $\tau_T\Psi(x)$ in $S_T$.
Periodic points for $T$ generate isolated periodic points in $S_T$ because, once the system falls in the periodic behavior, its future is fixed.

Preperiodicity in $T$ also implies equicontinuity in $T_H$, but $T_H$ may have other equicontinuous points.
The previously mentioned machine which always go to the left produces equicontinuous points for $T_H$ which are not equicontinuous nor preperiodic for $T_T$.

The following proposition states that the equicontinuity of preperiodic configurations is transmitted to factor subshifts of $T_H$, which will be helpful in the sequel.
\begin{prop}\label{p:shper}
 If $z\in S_H$ is a preperiodic word involving the machine head infinitely often, then it is isolated.
\end{prop}
\begin{proof}
 We can assume that $z$ is periodic, and then include the transient evolution in a larger ball.
 Let $p\in\Nset\setminus\{0\}$ be the period of $z$; let us prove that the ball $U=[z\scc0{\card Q\card A^{p+1}(p+1)^2}]_0$ of $S_H$ is equal to $\{z\}$.
 Let $z'\in U$ and $x\in\tau_H^{-1}(z')$. It can be seen that the head computing over $z'$ always remains between the positions $\ipart{-p/2}$ and $\ipart{p/2}$, which correspond to at most $\card Q\card A^{p+1}(p+1)$ distinct finite patterns. Hence there are $i<j\le\card Q\card A^{p+1}(p+1)$ such that $T^i(x)=T^j(x)$; as a consequence $\sigma^i(z')$ is $(j-i)$-periodic. Together with $\sigma^i(z)$, they are both $(j-i)p$-periodic and coincide on their first $(j-i)p$ letters, since $(j-i)p\le\card Q\card A^{p+1}(p+1)^2-i$. As a conclusion, $z'=z$.
\qed\end{proof} 
\subsection{Preperiodic machines}
When all the configurations are uniformly preperiodic, we say that the system is preperiodic, \ie there exist $q$, $p$ such that $T^{q+p}=T^q$.
In the present case, global preperiodicity of each of the considered systems comes directly from local preperiodicity of $T$; and it is equivalent to global equicontinuity of each of the systems as the next theorem establishes.
\begin{teo}
Considering a machine, the following statements are equivalent:
\begin{enumerate}
 \item\label{bounded} The head position is (uniformly) bounded.
\item Any configuration of $X$ (or $X_H$, $X_T$) is {\bf preperiodic}.
 \item\label{prep}$T$ (or $T_H$, $T_T$, $S_T$, $S_H$) is {\bf preperiodic}.
 \item\label{equi}$T$ (or $T_H$, $T_T$, $S_T$, $S_H$) is {\bf equicontinuous}.
 \item\label{finite} $S_T$ (or $S_H$) is {\bf finite}.
\end{enumerate}
\end{teo}
\begin{proof}We give only a sketch of the main implications.
\begin{itemize}
 \item It is quite obvious from the commutation diagrams that the preperiodicity of $T$, $T_H$ and $T_T$ are equivalent, and they imply those of $S_T$ and $S_H$.
They also imply, from Remark~\ref{r:eqpt}, that the head position is bounded.
 \item Clearly, the equicontinuity of $T$ and $T_H$ are equivalent.
 \item It is known from cellular automata theory that the equicontinuity of $T_H$, its preperiodicity, that of all its configuration and the finiteness of $S_H$ are equivalent.
 \item If the head position on all configurations is bounded, then from Remark~\ref{r:eqpt} they are all equicontinuous for $T_T$. $X_T$ being compact, $T_T$ is equicontinuous.
 \item It is obvious that $S_T$ is finite if and only if the head reads a bounded part of the initial configuration.
\qed
\end{itemize}
\end{proof}
\subsection{Sofic machines}
Now we allow computations where the head can go arbitrarily \xpr{far}, but without ever making \xpr{large} movements back.
\begin{defi}
 We say that a machine makes a \emph{right-cycle} (\emph{left-cycle}) of width $N\in\Nset$ over a configuration $x\in\az\times Q\times\Zset$ and a cell $i\in\Zset$ if there exist time steps $0=t_0<t_1<t_2$ such that the head position is $i$ at time $0$ and $t_2$, and is $i+N$ ($i-N$) at time $t_1$.
\end{defi}
In this section, we consider machines whose cycles have bounded width, \ie there exists an integer $N$ such that the machine cannot make any cycle wider than $N$.
These machines have been studied in~\cite{GajaJAC,Gaja07}, where it was proved that they are exactly the machines for which $S_T$ is sofic.
\begin{teo}~\label{t:sofic}
Considering a machine, the following statements are equivalent:
\begin{enumerate}
 \item\label{sofic} $S_T$ is {\bf sofic}.
 \item\label{i:tmheq} All configurations of $X_H$ that contain the head are {\bf equicontinuous}.
\end{enumerate}
\end{teo}
\begin{proof}~\begin{itemize}
 \item[\ref{sofic}$\Rightarrow$\ref{i:tmheq}]
 From~\cite{Gaja07}, we know that there exists an integer $N$ such that the machine cannot make any cycle wider than $N\in\Nset$, and let $x\in X_H$ a configuration containing the head within $\cc{-k}k$, for some $k\in\Nset$.
 Let us show that if $y\in[x\scc{-k-N}{k+N}]$, then for every $t\in\Nset$ we have $T_H^t(y)\in[T_H^t(x)\scc{-k}k]$.
 Let us remark that while the head is inside $\cc{-k-N}{k+N}$, we necessarily have $T_H^t(y)\in[T_H^t(x)\scc{-k}k]$.
 Let us suppose that there exists $j\in\Nset$ such that the head is outside $\cc{-k-N}{k+N}$ at time $j$ and let us take this $j$ minimal.
 Then the heads of $T_H^j(x)$ and $T_H^j(y)$ are outside $\cc{-k-N}{k+N}$.
 At some moment, the head has gone from $k$ to $k+N$ (or from $-k-N$); if it comes back to $\cc{-k}k$, it would make a cycle.
 Therefore, the head cannot come back to $\cc{-k}{k}$, and this is true both for $x$ and $y$, and we have the result.
 \item[\ref{i:tmheq}$\Rightarrow$\ref{sofic}]
  Conversely, assume that the head can do arbitrarily wide right-cycles in cell $0$, \ie for each $j\in\Nset$ there exists a cylinder $[u^j]_0$ of $X_H$ with $u^j\in (A\times Q)A^{n_j}$, with $n_j>j$, such that over each configuration of $[u^j]_0$, the head starts at $0$, it visits the whole interval $\cc0{n_j}$ and comes back to cell $0$.
 Let us take some configuration $c^j$ in each cylinder $[u^j]_0$.
 By compactness, the sequence $(c^j)_{j\in\Nset}$ admits an adhering value $c$, on which the head necessarily goes infinitely far to the right without ever coming back to cell $0$.
 By construction, for any $N$, there is some $j\in\Nset$ such that the configuration $c^j\scc{-N}N=c\scc{-N}N$.
 But there exists a time $t\in\Nset$ such that $T_H^t(c^j)$ has the head in cell $0$, whilst $T_H^t(c)$ has not; hence $c$ is not equicontinuous.
 \qed
\end{itemize}
\end{proof}
From \cite{GajaJAC}, any of the former properties implies that any configuration is either preperiodic or gives rise to a movement of the head arbitrarily far in some direction, but the converse is not true.
Any configuration of $S_H$ is preperiodic, hence this subshift is numerable.
%
%
\section{Bounded-zigzag machines}
Whilst the sofic machines did not allow any large cycle, we can wonder what happens when allowing a single one, or a finite number of these.
The first step is to allow one cycle of arbitrary width but to forbid two overlapped unbounded cycles (zigzags).
We remark that two independent cycles, each on a different direction, are allowed in this case.
\begin{defi}
 We say that a machine makes a \emph{right-zigzag} (resp., \emph{left-zigzag}) of width $N\in\Nset$ over a configuration $x\in\az\times Q\times\Zset$ and a cell $i\in\Zset$, if there exist time steps $0=t_0<t_1<t_{2}$ such that the machine position is $i$ at times $t_{0}$ and $t_2$, and $i+N$ (resp., $i-N$) at time $t_{1}$.
We say that a machine is \emph{bounded-zigzag} if the maximal width of the zigzags that it can make is finite.
\end{defi}
\subsection{Complexity of $S_T$}
While bounded cycle machines have a sofic shift $S_T$, the bounded-zigzag machines have a subshift recognized by a deterministic pushdown automata.
The words of $S_T$ contain information about the tape symbols and the head state.
From this data, it is possible to deduce the tape symbol of the visited cells and the relative position of the head at each time step.
In order to recognize $S_T$, we can register this information and check its coherence at each time step.
When the width of the cycles is bounded, we only need to register a finite part of the tape (bounded-cycle machines have a subshift that can be recognized by a finite state automaton).

When only one ``wide'' cycle can be done, we can register the information in a stack, from which it can be read exactly once (and is lost forever once read).
This corresponds to the fact that the cells registered in the stack cannot be visited any more and zigzags cannot be allowed. The complete proof can be found in the appendix.
\begin{teo}\label{t:zigstack}
A machine $T$ is bounded-zigzag if and only if $S_T$ is recognized by some deterministic pushdown automaton.
\end{teo}
\subsection{Complexity of $S_H$}
If we now adopt a point of view fixed on the tape --$S_H$-- rather than the head, a cycle in the subshift corresponds to a waiting time during which cell $0$ does not change.
We can adapt the previously built DPDA so that it recognizes exactly these waiting words between two visits of the head.
The key point here is that these languages are unary, and unary context-free languages are regular (see for example~\cite{Gins62}), and thus they can be recognized with a finite automaton.

When the machine is bounded-zigzag, the head can make at most one long cycle by side.
The rest of the time, the head is either moving closer to or farther from cell $0$, or staying in some finite window around cell $0$.
All of these behaviors can be recognized by a finite automaton, thus the language of $S_H$ is regular.
Therefore, we obtain a surprising reduction in language complexity when changing the point of view: if $S_T$ is recognized by some DPDA, then $S_H$ is sofic. The complete proof can be found in the appendix.
Note that, up to a rescaling of the tape alphabet, all factor subshifts can be reduced to the case of $S_H$.
\begin{teo}\label{t:zigsof}
 For any bounded-zigzag machine, all the factor subshifts of $T_H$ are sofic.
\end{teo}
The converse of this theorem is false: we can construct a machine with a tape with $n$ levels, where the head vertically shifts down the content of each level while moving right.
It rebounds when it finds a wall in the lowest level (which is erased in the same way), and does the same in the opposite direction.
We can see that the machine can make arbitrarily wide $n$-zigzags, each of independent length, in such a way that the factor subshifts of $T_H$ are sofic.

Nevertheless, we can prove that this kind of construction is possible only with a bounded $n$.
Let us introduce this formally.
\begin{defi}
 We say that a machine makes an \emph{$n$-cycle} of width $N\in\Nset$ over configuration $x\in\az\times Q\times\Zset$ and cell $i\in\Zset$, if there exist $2n+1$ time steps $0=t_0<t_1<\ldots<t_{2n}$ such that the head is in position $i$ at time $t_{2q}$ and outside $\cc{-N}N$ at time $t_{2q+1}$, for each $q\in\cc0n$.
We say that the machine is $n$-bounded-cycle if there is some $N$ such that the head cannot make $n$-cycles of width larger than $N$.
\end{defi}
When $S_T$ is sofic, the machine is $1$-bounded cycle.
Considering some machine $T$, we denote $\psi_N(x)\in\Nset\sqcup\{+\infty\}$ the maximum $n$ such that the machine can make an $n$-cycle of width $N$ over configuration $x$.
Clearly, $T$ is $n$-bounded cycle if and only if for some $N\in\Nset$, $\psi_N$ is bounded by $n-1$.

Let us call $\Phi_i(x)$ the set of time steps for which the head has position $i\in\Zset$ when computing over configuration $x$.
This set is linked to cycles by the following intuitive observation.
\begin{prop}\label{p:zzbound}
If $T$ is an $n$-bounded-cycle machine, then there exists $p\in\Nset$ such that for any cell $i\in\Zset$ and any non-preperiodic configuration $x\in X$, $\card{\Phi_i(x)}\le p$.
\end{prop}
\begin{proof}
Let $n,N\in\Nset$ be such that $max\sett{\psi_N(x)}{x\in X}=n-1$, and $x\in X$ such that $\card{\Phi_0(x)}>p=2n\card A^{2N+1}$ -- the case $i\ne0$ can be obtained by shifting.
Consider $\{t_0,\ldots,t_p\}\subset\Phi_0(x)$ with $t_0<t_1<\ldots<t_p$.
If we consider an $(n-1)$-cycle over $x$ in cell $0$, we can see that there exist $t_{k_1}<t_{k_2}<\ldots<t_{k_{n-1)}}$ such that for any $i\in\cc1{n-1}$, the head goes beyond $N$ or $-N$ between time steps $t_{k_i}$ and $t_{k_i+1}$, but not between (possibly equal) times $t_{k_i+1}$ and $t_{k_{i+1}}$.
This means that $t_{k_i}$ is the last time that the head is in $0$ before going beyond $\cc{-N}N$.
Let $k_0=-1$ and $k_{n}=p$, in such a way that $\cc0p=\bigcup_{i=0}^{n}I_i$, where $I_i=\cc{k_i+1}{k_{i+1}}$ for $0\le i\le n$. 
There are $n+1$ such intervals, so one of them, say $I_i$, has at least $\card A^{2N+1}$ elements; this is all the more the case for $\cc{t_{k_i+1}}{t_{k_{i+1}}}\supset\sett{t_{k_j}}{k_i<j\le k_{i+1}}$.
Hence, between time steps $t_{k_i+1}$ and $t_{k_{i+1}}$ there are at least $\card A^{2N+1}$ consecutive time steps in $\Phi_0(x)$ such that the head stays within the interval of cells $\cc{-N}N$.
As a result, there are $i,j\in\cc{t_{k_i+1}}{t_{k_{i+1}}}$ with $i<j$ and $T^i(x)=T^j(x)$, which implies that $x$ is preperiodic.
\qed\end{proof}
%
%
\begin{teo}
If $S_H$ is sofic, then $T$ is $n$-bounded-cycle for some $n$.
\end{teo}
\begin{proof}
Assume that $S_H=\tau_H(X_H)$ is recognized by some finite automaton with $N$ states, and that there exists some configuration $x\in X$ on which the machine makes some $N$-cycle of width $N$.
Let $t_0,\ldots,t_{2N}$ be as in the definition of $N$-cycles, and $u=\tau_H(x)\sco0{t_{2N}}$.
Let $o_0\ldots o_{t_{2N}+1}$ be the corresponding path of the finite automaton.
We can see that there are $i<j<N$ such that $o_{t_{2i}}=o_{t_{2j}}$, hence there is some periodic infinite word $z\in\tau_H(X_H)$ corresponding to the path $w$ that repeats the cycle $(o_{t_{2i}}\ldots o_{t_{2j}})$.
From Proposition~\ref{p:shper}, $z$ is isolated.
As a consequence, $w$ is the only path to start from $o_{t_{2i}}$.
Therefore, its vertices are all different, and $t_{2j}-t_{2i}\le N$, but in this case the head does not have the time to go beyond $\cc{-N}N$ between these two iterations, which is a contradiction.
We have proved that $T$ is $N$-bounded-cycle.
\qed
\end{proof}
Here, too, the converse is false, since it is easy to build a machine doing a given number of arbitrarily wide rebounds on specific wall characters before stopping.
The language of such a machine cannot be regular because the time intervals between two rebounds are not independent.
\subsection{Almost equicontinuity}
We have already seen that in sofic machines, almost all configurations of $X_H$ are equicontinuous.
It is still so when allowing $n$-cycles, though in this case there are some configurations with head which are not equicontinuous -- recall that Theorem~\ref{t:sofic} is an equivalence.
\begin{teo}
 If $T$ is an $n$-bounded-cycle machine for some $n$, then $T_H$ is almost equicontinuous.
\end{teo}
\begin{proof}
 By compactness of the space, it is enough to prove that for any cylinder $[u]$ and any $k\in\Nset$, there exist some $x\in[u]$ and some $m\in\Nset$ such that for any $y\in[x\scc{-m}m]$ and any $t\in\Nset$, $T_H^t(y)\in[T_H^t(x)\scc{-k}k]$.
 Let $N\in\Nset$ be as in the definition of $n$-bounded-cycle machine, $[u]$ a cylinder of $X_H$ and $k\in\Nset$.
 If $[u]$ contains some preperiodic configuration with the head, then we can easily find $m$ thanks to Remark~\ref{r:eqpt}.
 Otherwise, let us consider some configuration $x\in[u]$ (with the head) maximizing $\card{\Phi_{-k}(x)\sqcup\Phi_k(x)}$, which is finite thanks to Proposition~\ref{p:zzbound}.
 Let $m\in\Zset$ be such that $m\ge k$ and the interval $\cc{-m}m$ contains all the cells visited, when computing from $x$, up to time step $t=\max(\Phi_{-k}(x)\sqcup\Phi_k(x))$.
 Then we can see that any configuration $y\in[x\scc{-m}m]$ has the same evolution as $x$ until this time step, and that after that, its head cannot visit cell $-k$ nor $k$, otherwise it would contradict the maximality of $x$.
 We can deduce that the head of $x$ (then also $y$) is outside $\cc{-k}k$ after iteration $t$, otherwise it would be trapped between $-k$ and $k$ and would become periodic.
 We observe, then, that the cells of $\cc{-k}k$ evolve exactly in the same way for configurations $x$ and $y$.
\qed\end{proof}
The converse is untrue: imagine a machine whose head rebounds between two walls, each time shifting them to the left.
Every configuration where the head starts enclosed between two walls is equicontinuous.
Any finite pattern can be extended by adding walls to enclose the head, therefore equicontinuous points are dense, but the head can make an arbitrary number of arbitrarily wide cycles.
\section*{Conclusion}
The complexity of the Turing machine will always be very hard to understand.
In our attempt to treat this issue through the theories of topological and symbolic dynamics, we have found interesting relations between:
\begin{itemize}
 \item the head movements that can be observed during the computation;
 \item the density of equicontinuous points;
 \item the language complexity of the associated subshifts $S_T$ and $S_H$.
\end{itemize} 
These relations introduce a new point of view on how computation is performed.
In addition to generalizing them to more machines, the next step would be to study Turing machines as computing model by introducing a halting state, and to link all of these considerations to the result itself of the computation, and eventually the temporal or spatial complexity of the computation.

\bibliography{Xbib}

\newpage

\appendix\section*{Proofs}
The Ogden Lemma \cite{Odge} is a well-know generalization to the case of pushdown automata of the pumping lemma on finite automata. It can be expressed on paths of the graph as follows.
\begin{lem}\label{l:bomba}
 Consider a DPDA $(A,\Omega,\Gamma,\bot,\lambda,o_0,F)$, and $\pile0\ldots\pile n$ some path of its graph and $I\subset\cc0n$ a subset of distinguished positions of size $\card I>q=2^{\card\Omega^2\card\Gamma^2+1}$.
Then there exist four positions $0\le l_1\le l_2<l_3\le l_4\le n$ and such that:
\begin{enumerate}
 \item $\pilo{l_1}=\pilo{l_2}$;
 \item $\pilo{l_3}=\pilo{l_4}$;
 \item $\forall i\in\cc{l_1}{l_4},\length{\mu^{i}}\ge\length{\mu^{l_1}}$;
 \item $\forall i\in\cc{l_2}{l_3},\length{\mu^{i}}\ge\length{\mu^{l_2}}$.
 \item $\pile0\ldots\pile{l_1}\pile[\tilde]{l_2+1}\ldots\pile[\tilde]{l_3}\pile{l_4+1}\ldots\pile n$ is also a valid path of the graph, where $\tilde\mu^{t}=\mu^{l_1}\mu^{t}\soo{\length{\mu^{l_2}}}{\length{\mu^{t}}}$;
 \item $I\cap\co{l_2}{l_3}\ne\emptyset$;
 \item $\card{I\cap\co{l_1}{l_4}}\le q$;
 \item Either $I\cap\co0{l_1}\ne\emptyset\ne I\cap\co{l_1}{l_2}$ or $I\cap\co{l_3}{l_4}\ne\emptyset\ne I\cap\co{l_4}n$.
\end{enumerate}
\end{lem}

If $T$ is a machine with rule $\delta:A\times Q\to A\times Q\times\{-1,0,1\}$ and $\alpha,q\in A\times Q$, then we note $\delta_A(\alpha,q)=\beta$, $\delta_Q(\alpha,q)=p$ and $\delta_D(\alpha,q)=d$ if $\delta(\alpha,q)=(\beta,p,d)$.
If $u=(\overline{u},q,n)\in A^k\times Q\times \Zset$, then we can define the corresponding cylinder in space $X$:
\[[u]_i=\set{y,p,j}{A^\Zset\times Q\times \Zset}{y\in[\overline{u}]_i\et p=q\et (n\in \co{i}{i+\length u}\Rightarrow j=n)}~.\]
Let $\varepsilon$ denote the empty word.

Theorem~\ref{t:zigstack} comes from the following lemmas.

\begin{lem}\label{l:centerDFA}
Let $N$ be a fixed natural number and $T$ a Turing machine.
Given two partial configurations $u=(\overline{u},p,0),v=(\overline v,q,k)\in A^{2N+1}\times Q\times \cc{-N}N$, there exists a DFA $C_{u,v}$ that recognizes the language $\mathcal C_{u,v}$ of the words $(\tau_T\Psi(x))_{j=0}^t$ for $t\in\Nset$, $x\in [u]$ such that $T^t(x)\in [v]$ and for any $j\in \co0t$ the head position of $T^j(x)$ is in $\oo{-N}N$.

Moreover, if $x$ satisfies the conditions of $\mathcal C_{u,v}$, then every $y\in [u]$ also does, with the same time $t$.
\end{lem}

The language $\mathcal C_{u,v}$ can be either empty, a singleton or, when $v$ is periodic for $T$, infinite.
The automaton $C_{u,v}$ simply simulates $T$ by loading $u$ on its memory, and making the partial configuration over cells $\oo{-N}N$ evolve simply by applying the machine rule.
The next lemma corresponds to similar and more evolved proof.

\begin{lem}\label{l:LRcycleDPDA}
Let $N$ be a fixed natural number and let $T$ be a Turing machine that cannot do $1$-zigzags of width $N$.
If we have three partial configurations $u=(\overline{u},p,N),v=(\overline{v},q,k),u'=(\overline{u}',p',0)\in A^{2N+1}\times Q\times \Nset$ such that $u\scc{-N}0 = v\scc{-N}0$ and $\mathcal C_{u',u}\ne\emptyset$, then there exists a DPDA $R_{u,v}$ that recognizes the language $\mathcal R_{u,v}$ of the words $(\tau_T\Psi(x))_{j=0}^t$ for $t\in\Nset$, $x\in [u]$ such that $T^t(x)\in [v]$ and for any $j\in \co0t$, the head position of $T^j(x)$ is strictly positive.

Moreover, if $x$ satisfies the conditions of $\mathcal R_{u,v}$, then every $y$ such that $y\soi0=x\soi0$ also does, with the same time $t$.

Symetrically, if $u=(\overline{u},p,-N),v=(\overline{v},q,k),u'=(\overline{u}',p',0)\in A^{2N+1}\times Q\times \Zset_-$ such that $u\cc0N = v\cc0N$ and $\mathcal C_{u',u}\ne\emptyset$, then there exists a DPDA $L_{u,v}$ that recognizes the language $\mathcal L_{u,v}$ of the words $(\tau_T\Psi(x))_{j=0}^t$ for $t\in\Nset$, $x\in [u]$ such that $T^t(x)\in [v]$ and for any $j\in \co0t$, the head position of $T^j(x)$ is strictly negative.

Moreover, if $x$ satisfies the conditions of $\mathcal L_{u,v}$, then every $y$ such that $y\sio0=x\sio0$ also does, with the same time $t$.
\end{lem}
\begin{proof}
We will do the proof only for $\mathcal R_{u,v}$.
The automaton registers the states of the tape and updates them at each step.
The states of the cells at the right of the head will be registered in the internal state of the automaton, while the states of the cells at the left will be stocked in the stack.
The position of the head is given by the stack depth; in this way the head is always reading the symbol $w_0$.

We define actually an automaton in a slightly different model than previously defined. The initial and terminal states actually involve the content of the stack: we initially push a given finite word into the stack, and to accept a word, we verify if both the terminal state and the stack content are in some given finite sets. It is easy to see, by considering some complex encoding in the stack alphabet $\Gamma$, that this model can be simulated by the previous one.
The automaton $R_{u,v}$ has input alphabet $A\times Q$, states set $\Omega=(A^{\le N}\times Q)\sqcup\{REJECT\}$, stack alphabet $A^{\le N}$; its initial state is $o=(\overline{u}_N,p)$ and initial stack content the mirror of $\overline u\sco1N$; it terminates when the pair composed of the internal state and the stack content is in $F=\sett{((w,q),\mu)}{v=(\mu w,q,\length\mu+1)}$; its transition function $\lambda$ is defined by:
\[\lambda((\alpha,p),(w,q),\beta)=\soit{
((w\sco1N,\delta_Q(\alpha,q)),\delta_A(\alpha,p)\beta) &
\textrm{if }\both{p=q\et\delta_D(\alpha,p)=1\\w=\varepsilon\text{ or } w_0=\alpha}\\
((\beta\delta_A(\alpha,p)w\scc1{N-2},q'),\varepsilon) &
\textrm{if }\both{p=q\et\delta_D(\alpha,p)=-1\\w=\varepsilon\text{ or } w_0=\alpha}\\
REJECT & \textrm{in any other case.}
}\]


Let us denote by $\mu^j\in A^*$ and $(w^j,q^j)\in A^{\le N}\times Q$ the respectively stack content and internal state at iteration $j\in\Nset$.
\begin{itemize}
\item
We will prove by induction on $j\in\Nset$, that if $x$ satifies the conditions of $\mathcal R_{u,v}$, then $T^j(x) \in [(\mu^j w^j,q^j,\length{\mu^j}+1)]$.
\\
For $j=0$ it is clear, because $x\in [u] \subset [(\mu^0w^0,q^0,N)]$.
Let us suppose that it is true for a given $j\in\Nset$, and let us prove it for $j+1$.
\\
If $w^j\ne\varepsilon$, the head is reading the symbol $w^j_0$ and is in state $q^j$, hence the only input accepted is $(w^j_0,q^j)$.
In this case, the head will pass to state $\delta_Q(w^j_0,q^j)$ and will move to $\delta_D(w^j_0,q^j)$.
If $\delta_D(w^j_0,q^j)=1$ the automaton must push $\delta_A(w^j_0,q^j)$ and ``erase'' $w^j_0$.
If $\delta_D(w^j_0,q^j)=-1$ the automaton must replace $w^j_0$ by $\delta_A(w^j_0,q^j)$, pop a symbol and concatenate it to $w^j$.
\\
If $w^{j}=\varepsilon$, the automaton will accept $(\alpha,p)$ only if $p=q^j$; in this case it will work, assuming that $w^{j}=\alpha$.
\item
Now we need to prove that every word recognized by $R_{u,v}$ is in fact in $\mathcal R_{u,v}$.
We use recurrence to define the configuration $x$ that certifies this.
The first condition is that $x\in [(\overline{u}_N,p,N)]$, it follows from the first verification: $(\alpha,p) = (w^0_0,q^0) = (\overline{u}_N,p)$.
Let us suppose that we have defined $x=(\overline x,p,N)$ such that $T^s(x) \in [(\mu^s w^s,q^s,\length{\mu^s}+1)]$ for every $s\le j$ and that the set of cells visited by the head is $\cc ri$ for some $r<N$ and $i\ge\length{\mu^jw^j}+1$.
Let us prove that the same is true for $j+1$ for a suitable $x'$.
We can note that the condition $T^j(y) \in [(\mu^j w^j,q^j,\length{\mu^j}+1)]$ holds for any $y$ satisfying $y\scc0i=x\scc0i$.
\\
If $w^j=\varepsilon$, then the automaton will accept any pair $(\alpha,p)$ with $p=q^j$ if cell $k=\length{\mu^r}+1$ has already been visited; the value of $\overline{x}_k$ is important and cannot be defined to be $\alpha$. But if $k$ was visited, its value was registered in $w^s$ for somme $s$, and it has been erased because the head has moved to $k-N$ in some moment (then $k>N$).
The existence of $u' = (\overline{u}',p',0)$ such that $\mathcal C_{u',u} \ne\emptyset$ insures that the head has moved from $0$ to $N$, which means that, the head has made a $1$-zigzag to the right between cells $k-N$ and $k$, with is forbiden by hypothesis. 
Hence $k$ has not been visited before ($i<k$) and we can define $x'_k=\alpha$.
\\
When $w^j\ne\varepsilon$, we know that the value of cell $k$ in $T^j(x)$ is $w^j_0$. The automaton will only accept the pair $(\alpha,p) = (w^j_0,q^j)$.
This and the former construction insure that $T^{j+1}(x) \in [(\mu^{j+1} w^{j+1},q^{j+1},\length{\mu^{r+1}}+1)]$.
\qed\end{itemize}\end{proof}

\begin{proof}[of Theorem~\ref{t:zigstack}]

($\Rightarrow$)
Since $S_T$ does not regards the head position, we can suppose that the head starts at 0.
Let $x$ be a configuration.
\begin{itemize}
\item If the head does not exit the interval $\cc{-N}N$ during the whole evolution, then only $M$ is needed to recognize $\tau_T\Psi(x)$, we conclude that $\tau_T\Psi(x)\scc0k\in \mathcal C_{x\scc{-N}N,T^k(x)\scc{-N}N}$, for every $k\in \Nset$.
\item If the head exits $\cc{-N}N$ for the first time at iteration $t_0$, by the right side, and never comes back to cell $0$ after that, then $\tau_T\Psi(x)\scc0k\in \mathcal C_{x\scc{-N}N,T^{t_0}(x)\scc{-N}N}\mathcal R_{T^{t_0}(x)\scc{-N}N,T^{k}(x)\scc{-N}N}$, for any $k$.
\item If the head exits $\cc{-N}N$ for the first time at iteration $t_0$, by the right side, comes back to $0$ at iteration $t_1$, and never exit $\cc{-N}N$ again, then 
$\tau_T\Psi(x)\scc0k\in \mathcal C_{x\scc{-N}N,T^{t_0}(x)\scc{-N}N}\mathcal R_{T^{t_0}(x)\scc{-N}N,T^{t_1}(x)\scc{-N}N}\mathcal C_{T^{t_1}(x)\scc{-N}N,T^{k}(x)\scc{-N}N}$, for any $k$.
\item If the head exits $\cc{-N}N$ for the first time at iteration $t_0$, by the right side, comes back to $0$ at iteration $t_1$, and exits $\cc{-N}N$ again at $t_2$ and does not ever come back to $0$, then $\tau_T\Psi(x)\scc0k$ is in the concatenation of the languages $\mathcal C_{x\scc{-N}N,T^{t_0}(x)\scc{-N}N}$, $\mathcal R_{T^{t_0}(x)\scc{-N}N,T^{t_1}(x)\scc{-N}N}$, $\mathcal C_{T^{t_1}(x)\scc{-N}N,T^{t_2}(x)\scc{-N}N}$ and $\mathcal L_{T^{t_2}(x)\scc{-N}N,T^{k}(x)\scc{-N}N}$, for any $k$.
\item If the head exits $\cc{-N}N$ for the first time at iteration $t_0$, by the right side, comes back to $0$ at iteration $t_1$, exits $\cc{-N}N$ again at $t_2$, and comes back to $0$ at $t_3$, then $\tau_T\Psi(x)\scc0k$ is in the concatenation of $\mathcal C_{x\scc{-N}N,T^{t_0}(x)\scc{-N}N}$, $\mathcal R_{T^{t_0}(x)\scc{-N}N,T^{t_1}(x)\scc{-N}N}$, $\mathcal C_{T^{t_1}(x)\scc{-N}N,T^{t_2}(x)\scc{-N}N}$, $\mathcal L_{T^{t_2}(x)\scc{-N}N,T^{t_3}(x)\scc{-N}N}$ and $\mathcal C_{T^{t_3}(x)\scc{-N}N,T^{k}(x)\scc{-N}N}$, for any $k$.
\end{itemize}

The analogous case when the head first exits $\cc{-N}N$ through cell $-N-1$ can be treated in a similar way. We conclude that for any $x$ and any $k$, the word $\tau_T\Psi(x)\scc0k$ is in the language
\begin{eqnarray*}&&
\bigcup_{u^0,w}\mathcal C_{u^0,w} \bigcup
\bigcup_{u^0,v^0,w}\mathcal C_{u^0,v^0}\mathcal R_{v^0,w} \bigcup
\bigcup_{u^0,v^0,u^1,w}\mathcal C_{u^0,v^0}\mathcal R_{v^0,u^1}\mathcal C_{u^1,w} \bigcup\\&&
\bigcup_{u^0,v^0,u^1,b^1,w}\mathcal C_{u^0,v^0}\mathcal R_{v^0,u^1}\mathcal C_{u^1,b^1}\mathcal L_{b^1,w} \bigcup
\bigcup_{u^0,v^0,u^1,b^1,u^2,w}\mathcal C_{u^0,v^0}\mathcal R_{v^0,u^1}\mathcal C_{u^1,b^1}\mathcal L_{b^1,u^2}\mathcal C_{u^2,w} \bigcup\\&&
\bigcup_{u^0,b^0,w}\mathcal C_{u^0,b^0}\mathcal L_{b^0,w} \bigcup \bigcup_{u^0,b^0,u^1,w}\mathcal C_{u^0,b^0}\mathcal L_{b^0,u^1}\mathcal C_{u^1,w} \bigcup\\&&
\bigcup_{u^0,b^0,u^1,v^1,w}\mathcal C_{u^0,b^0}\mathcal L_{b^0,u^1}\mathcal C_{u^1,v^1}\mathcal R_{v^1,w} \bigcup
\bigcup_{u^0,b^0,u^1,v^1,u^2,w}\mathcal C_{u^0,b^0}\mathcal L_{b^0,u^1}\mathcal C_{u^1,v^1}\mathcal R_{v^1,u^2}\mathcal C_{u^2,w}~,
\end{eqnarray*}
where $u^i\in A^{2N+1}\times Q\times \{0\}$, $v^i\in A^{2N+1}\times Q\times \{N\}$, $b^i\in A^{2N+1}\times Q\times \{-N\}$, and $w\in A^{2N+1}\times Q\times \cc{-N}N$.
This language is recognizable by a DPDA since it is a concatenation and union of languages which are recognizable by DPDAs, thanks to Lemmas~\ref{l:centerDFA} and~\ref{l:LRcycleDPDA}.

We have to prove now that this union of languages contains only words of $\lang(S_T)$. The proof is similar for each of the listed languages; we will develop it only for
$\mathcal C_{u^0,v^0}\mathcal R_{v^0,u^1}\mathcal C_{u^1,b^1}\mathcal L_{b^1,w}$.

From Lemma~\ref{l:centerDFA}, we know that if $\mathcal C_{u^0,v^0}\ne\emptyset$, then any $x\in [u^0]$ will satisfy $\tau_T\Psi(x)\scc0{t_0}\in \mathcal C_{u^0,v^0}$ if the head position at time $t_0$ is $N$.
From Lemma~\ref{l:LRcycleDPDA}, if $\mathcal R_{v^0,u^1}\ne\emptyset$, then there exists $y\in [v^0]$ and $t_1$ such that $T^{t_1}(y)\in [u^1]$ and $\tau_T\Psi(x)\scc{t_0}{t_1}\in \mathcal R_{v^0,u^1}$.
We define $x\sci{N}=y\sci N$, which will satisfy $\tau_T\Psi(x)\scc0{t_1}\in \mathcal C_{u^0,v^0}\mathcal R_{v^0,u^1}$.
From the same lemmas, we know that the values of $x$ on $\ic{-N}$ are still ``free'' and $T^{t_1}(x)\in [u^1]$ gives $\tau_T\Psi(x)\scc{t_1}{t_2}\in \mathcal C_{u^1,b^1}$, where $t_2$ is the instant in which the head reaches the cell $-N$ for the first time.

We can suppose that $\mathcal L_{b^1,w}$ is not empty -- otherwise the result is trivial. Then there exists $y'$ such that $T^{k}(y')\in [w]$ and $\tau_T\Psi(x)\scc{t_2}{k}\in \mathcal L_{b^1,w}$.
The values of $y'$ over $\oi{-N}$ are not important and we can fix them to those of $T^{t_2}(x)$, or in other words, to define $x\sic{-N}=y\sic{-N}$. We obtain
$$\tau_T\Psi(x)\scc0k\in\mathcal C_{u^0,v^0}\mathcal R_{v^0,u^1}\mathcal C_{u^1,b^1}\mathcal L_{b^1,w}~.$$
This completes the proof.

($\Leftarrow$)
Let us assume that the language of $S_T$ is recognized by some DPDA $M$, that $p$ is as in Lemma~\ref{l:bomba}, and that the machine can do a $1$-zigzag of width $N=p+3$;  we can easily find some configuration $x$ with time steps $0<t_1<t_2<t_3$ such that the machine visits cell $1$ at time $0$, cell $N$ at times $t_1$ and $t_3$, and cell $0$ at time $t_2$.
It can also be assumed that the zigzag is minimal, in the sense that no other configuration satisfies the condition with a lower $t_3$. Moreover, we can assume that $t_1$ is the last time when cell $N$ is visited before $t_3$, and $t_2$ is the first time when cell $0$ is visited. Note that $t_2-t_1\ge N$. Let $c=\pile0\ldots\pile{t_3}$ the corresponding path in the graph of $M$.

The key point of the proof is that, thanks to the determinism of the automaton, given $w \in\lang(S_T)$, the $i-th$ cell is visited by the head for the fist time if and only if the corresponding vertex in the graph of $M$ has out-degree more than $1$.
Since the out-degree of a vertex $(q,u)$ of $M$ depends only on $(q,u_0)$.
Let $V$ be the set of vertices with out-degree $1$, and $L$ be the subset of $V$ corresponding to vertices whose unique out-neighbor is not in $V$ and such that this unique transition corresponds to a left movement of the head. These vertices represent cells which are at the left extremity of some visited zone.
For instance, note that the vertices $\pile{t_1+1},\ldots,\pile{t_2-1}$ are in $V$ since the corresponding visited cells are between $1$ and $N$, and $\pile{t_2-1}$ is the first vertex of the path $c$ to belong to $L$.

If we apply Lemma~\ref{l:bomba} with $I=\oo{t_1}{t_2-1}$, we obtain time steps $0\le l_1\le l_2<l_3 \le l_4<t_3$ such that 
$$\tilde c=\pile0\ldots\pile{l_1}\pile[\tilde]{l_2+1}\ldots\pile[\tilde]{l_3}\pile{l_4+1}\ldots\pile{t_3}$$
is a valid path in the graph of $M$, \ie it can be obtained from some configuration $\tilde x$, which we can suppose to have the head in cell $1$ without loss of generality.

\begin{itemize}
\item First, suppose $l_4\ge t_2-1$. 
Since $\card{I\cap\oc{l_1}{l_4}}\le p$, we must have $l_1>t_1$.
Moreover, the nonemptiness of $I\cap\oc{l_2}{l_3}$ gives $t_1<l_1\le l_2<t_2-1\le l_4\le t_3$.
The vertices of $d=\pile{l_1}\ldots\pile{t_2-1}$ are in $V$, then $d$ is the only subpath of this length starting at $\pile{l_1}$.
Thus, $\pile{l_1}\pile[\tilde]{l_2+1}\ldots\pile[\tilde]{t_2-1-l_2+l_1}=d$.
In particular, $\pile[\tilde]{t_2}=\pile{t_2-l_2+l_1}$.
From the lemma, $\pilo[\tilde]{t_2}=\pilo{t_2}$; the same automaton rule is applied in both vertices, and since $\pile{t_2}$ is not in V, we conclude that  $\pile{t_2-l_2+l_1}\not \in V$.
This results in $l_1=l_2$, and from the Ogden Lemma we get $l_3<l_4<t_2$, which is a contradiction.

\item Now suppose that $l_4<t_2$. As no vertex of $\tilde c$ is in $L$ before $\pile{t_2-1}$, we can see that, in this path too, the vertex $\pile{t_2}$ corresponds, at time $t_2-l_4+l_3-l_2+l_1$, to the first visit of cell $0$ -- at the first time we go more to the left than the visited zone.
Since both paths coincide after that, the head does the same movements, and we obtain that its position at the last vertex $\pile{t_3}$ of path $\tilde c$ is $N$. But, from path $c$, we know that $\pile{t_3}\in V$, so on $\tilde x$ too the machine had already visited cell $N$ before arriving on this vertex. It could not be after time $t_2-l_4+l_3-l_2+l_1$, since from then on we have followed the same positions as in $c$, hence $\tilde c$ represents a $1$-zigzag; from the last point of the lemma, it is shorter than $c$, so $\tilde x$ satisfies the construction hypotheses of $x$ but contradicts its minimality.
\qed
\end{itemize}
\end{proof}


\begin{proof}[of Theorem~\ref{t:zigsof}]
Let us define the following languages.

$$\overline{\mathcal R_{u,v}}=\sett{u_0^k}{\exists w\in \mathcal R_{u,v},\length w=k}$$
$$\overline{\mathcal L_{u,v}}=\sett{u_0^k}{\exists w\in \mathcal L_{u,v},\length w=k}$$

It is a context-free language since it is the transformation of a context-free language through a letter morphism.
It is also a regular language because it uses a single symbol $u_0$.
If $x$ and $t$ satisfy the conditions of $\mathcal R_{u,v}$, then $\tau_H\Phi(x)\scc0{t-1}\in\overline{\mathcal R_{u,v}}$.

We also define the language $\mathcal C_{u,v}$ of the words $\tau_H\Phi(x)\soo0t$ with $t\in\Nset$ and $x\in [u^0]$ such that $T^t(x)\in [u^1]$ and for any $j\in \co0t$, the head of $T^j(x)$ is in $\oo{-N}N$.

It is recognized by an automaton that simulates $M$ and accepts a pair $(\alpha,p)$ if and only if the current head position is $0$, and $p$ and $\alpha$ match the simulation.

If the head starts at cell $0$, the analogous concatenation and union of the $\overline{\mathcal C}$s, $\overline{\mathcal R}$s and $\overline{\mathcal L}$s would represent $\lang(S_H)$. But if the head does not start at $0$, we need to consider, for $u=(\overline{u},p,0)\in A^{2N+1}\times Q\times \{0\}$, the language $\mathcal B_u$ of the words $\overline{u}_0^{t}$ for which there exists $x$ with $T^t(x)\in [u]$ and for any $j<t$, the head of $T^j(x)$ is not in cell $0$.
$\mathcal B_u$ represents the set of sequences of states observed at cell $0$ until the head reaches it, when the partial configuration $u$ is observed in $\cc{-N}N$.
$\mathcal B_u$ is always an nonempty ``interval'', \ie $\mathcal B_u=\sett{u_0^t}{0\le t\le n}$ for some $n\in\Nset$ which may be $0$ -- if $u$ is a ``garden of Eden''.


Since $\mathcal B_u$ is either finite or equal to $\sett{u_0^t}{t\in \Nset}$, it can be recognized with a DFA $B_u$.
$\lang(S_H)$ will be the concatenation and union of $\mathcal B$s and the other languages.

Globally, we obtain that $\tau_F\Phi(x)\scc0k$ is in the following union:
\begin{eqnarray*}
&& \{x_0^k\}\bigcup
\bigcup_{u^0,w} \mathcal B_{u^0}\overline{\mathcal C_{u^0,w}} \bigcup
\bigcup_{u^0,v^0,w}\mathcal B_{u^0}\overline{\mathcal C_{u^0,v^0}}\overline{\mathcal R_{v^0,w}} \bigcup
\bigcup_{u^0,v^0,u^1,w}\mathcal B_{u^0}\overline{\mathcal C_{u^0,v^0}}\overline{\mathcal R_{v^0,u^1}}\overline{\mathcal C_{u^1,w}} \bigcup\\
& & \bigcup_{u^0,v^0,u^1,b^1,w}\mathcal B_{u^0}\overline{\mathcal C_{u^0,v^0}}\overline{\mathcal R_{v^0,u^1}}\overline{\mathcal C_{u^1,b^1}}\overline{\mathcal L_{b^1,w}} \bigcup
\bigcup_{u^0,v^0,u^1,b^1,u^2,w}\mathcal B_{u^0}\overline{\mathcal C_{u^0,v^0}}\overline{\mathcal R_{v^0,u^1}}\overline{\mathcal C_{u^1,b^1}}\overline{\mathcal L_{b^1,u^2}}\overline{\mathcal C_{u^2,w}} \bigcup\\
& & \bigcup_{u^0,b^0,w}\mathcal B_{u^0}\overline{\mathcal C_{u^0,b^0}}\overline{\mathcal L_{b^0,w}} \bigcup
\bigcup_{u^0,b^0,u^1,w}\mathcal B_{u^0}\overline{\mathcal C_{u^0,b^0}}\overline{\mathcal L_{b^0,u^1}}\overline{\mathcal C_{u^1,w}} \bigcup\\
& & \bigcup_{u^0,b^0,u^1,v^1,w}\mathcal B_{u^0}\overline{\mathcal C_{u^0,b^0}}\overline{\mathcal L_{b^0,u^1}}\overline{\mathcal C_{u^1,v^1}}\overline{\mathcal R_{v^1,w}} \bigcup
\bigcup_{u^0,b^0,u^1,v^1,u^2,w}\mathcal B_{u^0}\overline{\mathcal C_{u^0,b^0}}\overline{\mathcal L_{b^0,u^1}}\overline{\mathcal C_{u^1,v^1}}\overline{\mathcal R_{v^1,u^2}}\overline{\mathcal C_{u^2,w}}~.
\qed
\end{eqnarray*}
\end{proof}

\end{document}